\documentclass[a4paper, 11pt]{article}

\usepackage{amsmath}
\usepackage{amsthm}
\usepackage{amssymb}
\usepackage{ifpdf}
\usepackage{graphicx}
\usepackage{xcolor}
\usepackage{algorithm}
\usepackage[]{algpseudocode}
\usepackage{subfigure}
\usepackage[subnum]{cases}

\usepackage{geometry}
\newtheorem{theorem}{Theorem}
\newtheorem{lemma}[theorem]{Lemma}

\newcommand{\fig}[3]{
	\begin{figure}[ht]
		\centering
		\includegraphics{#3}
		\caption{#2}
		\label{#1}
	\end{figure}
}

\newcommand{\p}[2]{\{#1, \ldots, #2\}}
\newcommand{\pij}{\p{i}{j}}

\newcommand{\Matching}{\textsc{Matching}}

\begin{document}

\title{Faster Bottleneck Non-crossing Matchings of Points in Convex Position}
\date{}

\author{
	Marko Savi\'{c}\footnote{University of Novi Sad, Faculty of Sciences, Department of Mathematics and Informatics. Partly supported by Ministry of Education and Science, Republic of Serbia. {\tt \{marko.savic, milos.stojakovic\}@dmi.uns.ac.rs}}
	\and
	Milo\v{s} Stojakovi\'{c}\footnotemark[1] \footnote{Partly supported by Provincial Secretariat for Science, Province of Vojvodina.}
}

\maketitle

\begin{abstract}
	Given an even number of points in a plane, we are interested in matching all the points by straight line segments so that the segments do not cross. Bottleneck matching is a matching that minimizes the length of the longest segment. For points in convex position, we present a quadratic-time algorithm for finding a bottleneck non-crossing matching, improving upon the best previously known algorithm of cubic time complexity.
\end{abstract}

\section {Introduction}

Let $P$ be a set of $n$ points in the plane, where $n$ is an even number. Let $M$ be a perfect matching of points in $P$, using $n/2$ straight line segments to match the points, that is, each point in $P$ is an endpoint of exactly one line segment. We forbid line segments to cross. Denote the length of a longest line segment in $M$ with $bn(M)$, which we also call the \emph{value} of $M$. We aim to find a matching that minimizes $bn(M)$. Any such matching is called \emph{bottleneck matching} of $P$.

\subsection{Related work}

There is plentiful research on various geometric problems involving pairings without crossings. Some of considered problems examine matchings of various planar objects, see \cite{aloupis2013non, aloupis2015matching, kratochvil2013non}, while more basic problems involve matching pairs of points by straight line segments, see \cite{aichholzer2010edge, aichholzer2009compatible, alon1993long}. There is always a non-crossing matching of points with non-crossing segments, and moreover it is straightforward to prove that a matching minimizing the total sum of lengths of its segments has to be non-crossing.

In \cite{chang1992solving}, Chang, Tang and Lee gave an $O(n^2)$-time algorithm for computing a bottleneck matching of a point set, but allowing crossings. This result was extended by Efrat and Katz in \cite{efrat2000computing} to higher-dimensional Euclidean spaces.

Abu-Affash, Carmi, Katz and Trablesi showed in \cite{abu2014bottleneck} that the problem of computing non-crossing bottleneck matching of a point set is NP-complete and does not allow a PTAS. They gave a $2\sqrt{10}$ factor approximation algorithm, and also showed that the case where all points are in convex position can be solved exactly in $O(n^3)$ time. In \cite{abu2015approximating}, Abu-Affash, Biniaz, Carmi, Maheshwari and Smid presented an algorithm for computing a non-crossing bottleneck plane matching of size at least $n/5$ in $O(n \log^2 n)$ time. They then extended it to provide an $O(n \log n)$-time approximation algorithm which computes a plane matching of size at least $2n/5$ whose edges have length at most $\sqrt{2}+\sqrt{3}$ times the length of a longest edge in a non-crossing bottleneck matching.

Bichromatic (sometimes also called bipartite) versions of the bottleneck matching problem, where only points of different colors are allowed to be matched, have also been studied. Efrat, Itai and Katz showed in \cite{efrat01geometryhelps} that a bottleneck matching between two point sets, with possible crossings, can be found in $O(n^{3/2}\log n)$ time. Bichromatic non-crossing bottleneck problem was proved to be NP-complete by Carlson, Armbruster, Bellam and Saladi in \cite{carlsson2015bottleneck}.

Biniaz, Maheshwari and Smid in \cite{biniaz2014bottleneck} study special cases of non-crossing bichromatic bottleneck matchings. They show that the case where all points are in convex position can be solved in $O(n^3)$ time with an algorithm similar to the one for monochromatic case presented in \cite{abu2014bottleneck}. They also consider the case where the points of one color lie on a line and all points of the other color are on the same side of that line, providing an $O(n^4)$ algorithm to solve it. The same results for these special cases are independently obtained in \cite{carlsson2015bottleneck}. In \cite{biniaz2014bottleneck} an even more restricted problem, a case where all points lie on a circle, is solved by constructing an $O(n \log n)$-time algorithm.

\subsection{Monochromatic bottleneck non-crossing matchings for convex point sets and our results}

In what follows we consider the case where all points of $P$ are in convex position, i.e.\ they are the vertices of a convex polygon $\mathcal{P}$, and they are monochromatic, i.e.\ any two points from $P$ can be matched.
As we are going to deal with matchings without crossings, from now on, the word matching is used to refer only to pairings that are crossing-free.

Let us label the points $v_0, v_1, \ldots, v_{n-1}$ in positive (counterclockwise) direction. To simplify the notation, we will often use only the indices when referring to the vertices. We write $\pij$ to represent the sequence $i, i+1, i+2, \ldots, j-1, j$, where all operations are calculated modulo $n$; note that $i$ is not necessarily less than $j$, and $\pij$ is not the same as $\p{j}{i}$. We say that $(i,j)$ is a \emph{feasible} pair if there exists a matching containing $(i,j)$, which in this case simply means that $\pij$ is of even size.

The problem of finding a bottleneck matching of points in convex position can be solved in polynomial time using dynamic programming algorithm, as presented in \cite{abu2014bottleneck}. Similar algorithm for bichromatic case is presented in \cite{biniaz2014bottleneck} and \cite{carlsson2015bottleneck}. The algorithm is fairly straightforward, and we are going to describe it briefly.

The subproblems we consider are the tasks of optimally matching only the points in $\pij$, where $i, j\in\p{0}{n-1}$ and $j-i$ is odd. Each matching $M$ on $\pij$ matches $i$ with some $k \in \p{i+1}{j}$, where $(i,k)$ is feasible. Segment $(i,k)$ divides $M$ in two parts, a matching on $\p{i+1}{k-1}$ and a matching on $\p{k+1}{j}$. If we solve those two parts optimally, we can combine them into an optimal matching of $\pij$ that contains $(i,k)$. We go through all the possibilities for $k$ and take the best matching obtained in this way, yielding an optimal matching of points in $\pij$. If we denote the value of this optimal matching by $b_{i,j}$, we get the following recursive formula,

\[
b_{i,j} = \min_{k = i+1, i+3, \ldots, j}
\begin{cases}
	|v_iv_j| & \text{if } j-i = 1 \\
	\max\{|v_iv_k|, b_{k+1,j}\} & \text{if } k-i = 1 \\
	\max\{|v_iv_k|, b_{i+1,k-1}\} & \text{if } k = j \\
	\max\{|v_iv_k|, b_{i+1,k-1}, b_{k+1,j}\} & \text{otherwise}.
\end{cases}
\]

This formula is then used to to fill in the dynamic programming table. There are $O(n^2)$ entries, and to calculate each we need $O(n)$ time. Therefore, the described algorithm finds a bottleneck matching for monochromatic points in convex position in $O(n^3)$ time.

In this paper, we present a faster algorithm for finding a bottleneck matching for monochromatic points in convex position, with only $O(n^2)$ time complexity. En route, we prove a series of results that give insights in the properties and structure of bottleneck matchings.

\section{Structure of bottleneck matching}

Our aim is to show the existence of a bottleneck matching with a certain structure that we can utilize to construct an efficient algorithm. We do so by proving a sequence of lemmas, with each lemma imposing an increasingly stronger condition on the structure.

Let us split all point pairs into the two categories. Pairs consisting of two neighboring vertices of $\mathcal{P}$ are called \emph{edges}, and all other pairs are called \emph{diagonals}. Each matching is, thus, comprised of edges and diagonals.

The \emph{turning angle} of $\pij$, denoted by $\tau(i,j)$, is the angle by which the vector $\overrightarrow{v_iv_{i+1}}$ should be rotated in positive direction to align with the vector $\overrightarrow{v_{j-1}v_j}$, see Figure~\ref{fig:TurningAngle}.

\fig{fig:TurningAngle}{Turning angle.}{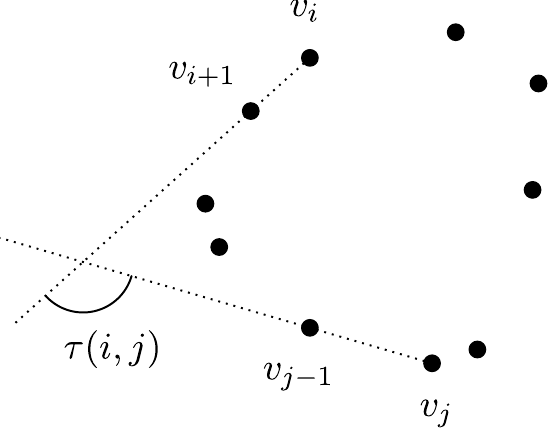}

\begin{lemma}
	\label{lem:PiHalf}
	There is a bottleneck matching $M$ of $P$ such that all diagonals $(i,j) \in M$ have $\tau(i,j) > \pi/2$.
\end{lemma}
\begin{proof}
	Suppose there is no such matching. Let $M'$ be a bottleneck matching with the least number of diagonals. By assumption, there is a diagonal $(i,j) \in M'$ such that $\tau(i,j) \leq \pi/2$, see Figure~\ref{fig:PiHalf1}. If we replace all pairs from $M'$ lying in $\pij$ with edges $(i,i+1), (i+2,i+3), \ldots, (j-1,j)$, we obtain a new matching $M^*$, see Figure~\ref{fig:PiHalf2}. The diameter of $\pij$, i.e.\ the longest distance between any pair of points from $\pij$, is achieved by the pair $(i,j)$, so $bm(M^*) \leq bm(M')$. Since $M'$ is a bottleneck matching, $bm(M^*) = bm(M')$, meaning that $M^*$ is a bottleneck matching as well. Diagonal $(i,j)$ belongs to $M'$, but does not belong to $M^*$, so the new matching, $M^*$, has at least one diagonal less than $M'$, which contradicts the assumption.
\end{proof}

	\begin{figure}[ht]
		\centering
		\subfigure[]{
			\label{fig:PiHalf1}
			\includegraphics{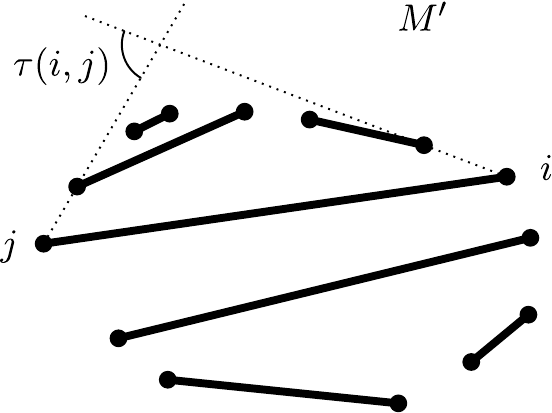}
		}
		\hspace{20pt}
		\subfigure[]{
			\label{fig:PiHalf2}
			\includegraphics{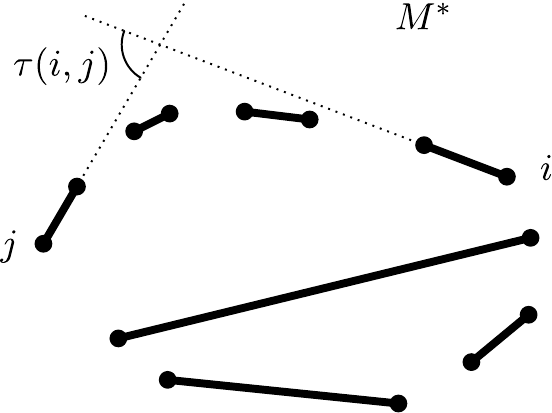}
		}
		\caption{Matchings before ($M'$) and after ($M^*$) the transformation.}
		\label{fig:PiHalf}				
	\end{figure}

Let us consider the division of the polygon $\mathcal{P}$ into regions obtained by cutting it along diagonals (but not edges) of the given matching $M$. Each region in this division is bounded by some diagonals of $M$ and by some edges from the polygon's boundary. If there are exactly $k$ diagonals bounding a region, we say the region is \emph{$k$-bounded}. Any maximal sequence of diagonals connected by $2$-bounded regions is called a \emph{cascade}, see Figure~\ref{fig:Cascades}.

\fig{fig:Cascades}{Diagonals inside each shaded area make a single cascade. There are three cascades with only one diagonal, one cascade with two diagonals, and one cascade with three diagonals.}{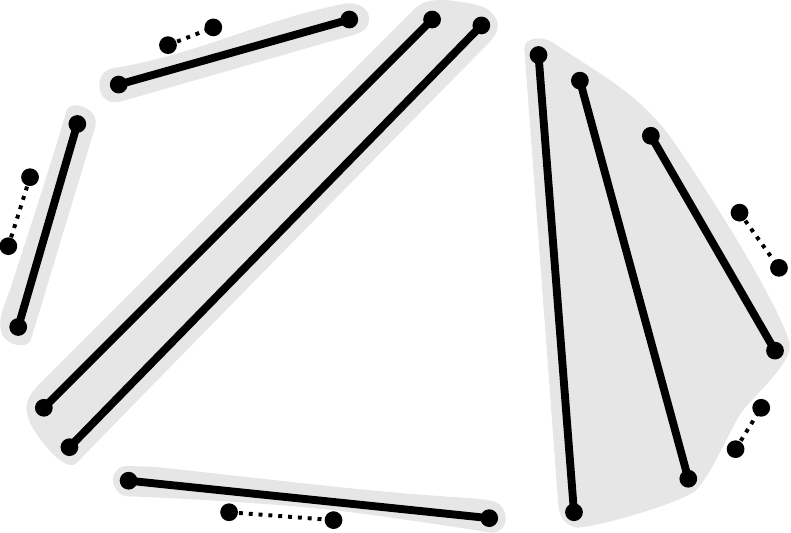}

\begin{lemma}
	\label{lem:ThreeCascades}
	There is a bottleneck matching having at most three cascades.
\end{lemma}
\begin{proof}
	Let $M$ be a matching provided by Lemma~\ref{lem:PiHalf}, with turning angles of all diagonals greater than $\pi/2$. There cannot be a region bounded by $4$ or more diagonals of $M$, since if it existed, the total turning angle would be greater than $2\pi$. Hence, $M$ only has regions with at most $3$ bounding diagonals. Suppose there are two or more $3$-bounded regions. We look at arbitrary two of them. There are two diagonals bounding the first region and two diagonals bounding the second region such that these four diagonals are in cyclical formation, meaning that each diagonal among them has other three on the same side. Applying the same argument once again we see that this situation is impossible because it yields turning angle greater than $2\pi$. We conclude that there can be at most one $3$-bounded region. 
\end{proof}

The case of a bottleneck matching having exactly three cascades is possible, as shown in Figure~\ref{fig:ThreeCascadesExample}.

\fig{fig:ThreeCascadesExample}{Configuration of points for which the only bottleneck matching has exactly three cascades.}{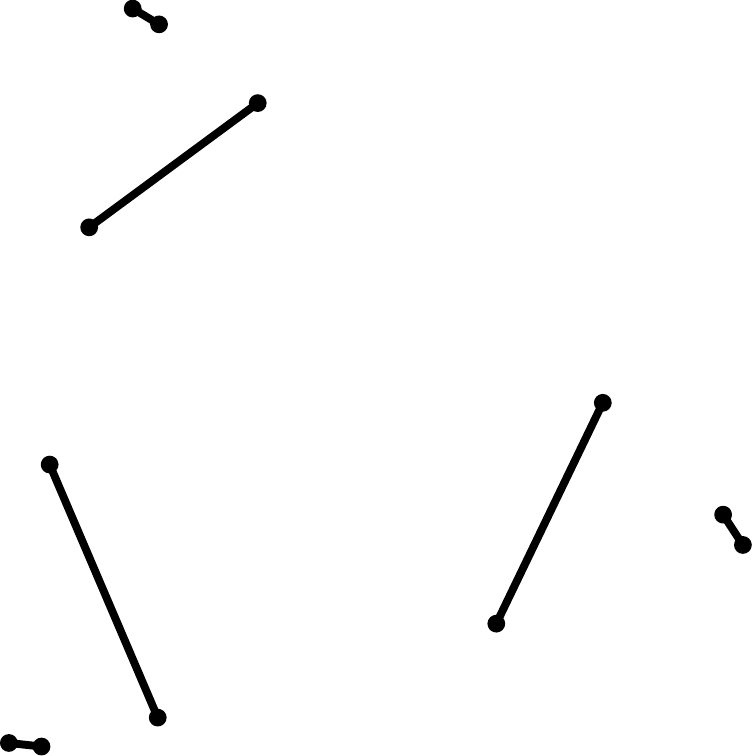}

Obviously, it is not possible for a matching to have exactly two cascades. So, from Lemma~\ref{lem:ThreeCascades} we know that there is a bottleneck matching which either has at most one cascade and no $3$-bounded regions, or it has a single $3$-bounded region and exactly three cascades. In the following section we define a set of simpler problems that will be used to find an optimal solution in both of these cases.

\section{Matchings with at most one cascade}
\label{sec:Subproblems}

When talking about matchings with minimal value under certain constraints, we will refer to these matchings as \emph{optimal}.

For $j-i$ odd, let $\Matching(i,j)$ be the problem of finding an optimal matching $M_{i,j}$ of points $\pij$, so that $M_{i,j}$ has at most one cascade, and the segment $(i,j)$ belongs to a region bounded by at most one diagonal from $M_{i,j}$ different from $(i,j)$.

If $j-i = 1$, then the solution to $\Matching(i,j)$ is exactly the edge $(i,j)$. If $j-i > 2$, we consider the following cases. If there is a solution to $\Matching(i,j)$ that contains the pair $(i,j)$, then $M_{i,j}$ can be constructed by taking $(i,j)$ together with $M_{i+1,j-1}$ . If not, then at least one of the edges $(i,i+1)$ and $(j-1,j)$ must be a part of $M_{i,j}$ (as otherwise points $i$ and $j$ would be endpoints of two different diagonals from $M_{i,j}$, neither of which is $(i,j)$), which is not allowed (by the requirement that the region containing $(i,j)$ has at most one other bounding diagonal). If $(i,i+1) \in M_{i,j}$, then $M_{i,j}$ can be constructed from $M_{i+2,j}$ and the edge $(i,i+1)$. Similarly, if $(j-1,j) \in M_{i,j}$, then we can get $M_{i,j}$ as $M_{i,j-2}$ plus the edge $(j-1,j)$.

Since these problems have optimal substructure, we can apply dynamic programming to solve them. If $bn(M_{i,j})$ is saved into $S[i,j]$, the following recurrent formula can be used to calculate the solution to $\Matching(i,j)$ for all feasible pairs $(i,j)$,
\begin{numcases}{S[i,j] = \min}
	\max\{S[i+1,j-1], |v_iv_j|\}\label{eqn:sa} \\
	\max\{S[i+2,j], |v_iv_{i+1}|\}\label{eqn:sb} \\
	\max\{S[i,j-2], |v_{j-1}v_j|\}.\label{eqn:sc}
\end{numcases}
Initially, we set $S[i,i] = 0$, for all $i$, and then we fill values in $S$ in order of increasing $j-i$, so that all subproblems are already solved when needed.

Beside the value of a solution to $\Matching(i,j)$, it is going to be useful to determine if pair $(i,j)$ is necessary for constructing $M_{i,j}$, i.e.\ we want to know if all the solutions to $\Matching(i,j)$ contain $(i,j)$. If that is true then we call such a pair \emph{necessary}. This can be easily incorporated into the calculation of $S[i,j]$. Namely, if case~(\ref{eqn:sa}) is the only one achieving minimum among cases (\ref{eqn:sa}), (\ref{eqn:sb}) and (\ref{eqn:sc}), we set $necessary(i,j)$ to $\top$, otherwise we set it to $\bot$.

We have $O(n^2)$ subproblems, each of which takes $O(1)$ time to be calculated. Hence, all calculations together require $O(n^2)$ time and the same amount of space. To find the optimum for matchings with at most one cascade, we just find the minimum of all $S[i+1,i]$, and take any $M_{i+1,i}$ that achieves it. This step takes only linear time.

Note that we calculated only the values of solutions to all subproblems. If an actual matching is needed, it can be easily reconstructed in linear time from the data in $S$.

\section{Finding bottleneck matching}
\label{sec:FindingBottleneckMatching}

As we concluded earlier, there is a bottleneck matching of $P$ having either at most one cascade, or exactly three cascades. An optimal matching with at most one cascade can be found easily from calculated solutions to subproblems, as shown in the previous section. Next, we focus on finding an optimal matching among all matchings with exactly three cascades, denoted by \emph{$3$-cascade matchings} in the following text.

Any three distinct points $i$, $j$ and $k$, where $(i,j)$, $(j+1, k)$ and $(k+1, i-1)$ are feasible pairs, can be used to construct a $3$-cascade matching by simply taking a union of $M_{i,j}$, $M_{j+1,k}$ and $M_{k+1,i-1}$. To find the best one we could run through all possible triplets $(i,j,k)$ and see which one minimizes $\max\{S[i,j], S[j+1,k], S[k+1,i-1]\}$. However, that requires $O(n^3)$ time, and thus is not suitable, since our goal is to design a faster algorithm. Our approach is to show that instead of looking at all $(i,j)$ pairs, it is enough to select $(i,j)$ from a set of linear size, which would reduce the search space to quadratic number of possibilities, so the search would take only $O(n^2)$ time.

In a $3$-cascade matching, let us call the three diagonals bounding the single $3$-bounded region the \emph{inner} diagonals.

\begin{lemma}
	\label{lem:BottleneckWithAllNecessary}
	If there is no bottleneck matching with at most one cascade, then there is a bottleneck $3$-cascade matching whose every inner diagonal is necessary.
\end{lemma}
\begin{proof}
	Take any $3$-cascade bottleneck matching $M$. If it has an inner diagonal $(i,j)$ that is not necessary, then (by definition) there is a solution to $\Matching(i,j)$ that does not contain the pair $(i,j)$ and has at most one cascade. We use that solution to replace all pairs from $M$ that are inside $\pij$, and thus obtain a new $3$-cascade matching that does not contain the pair $(i,j)$. Since $M$ was optimal and there was at most one cascade inside $\pij$, replaced pairs were also a solution to $\Matching(i,j)$, so the new matching must have the same value as the original matching. And since there is no bottleneck matching with at most one cascade, the new matching must be a bottleneck $3$-cascade matching as well. We repeat this process until all inner diagonals are necessary. The process has to terminate because the $3$-bounded region is getting larger with each replacement.
\end{proof}

We say that $(i,j)$ is a \emph{candidate} diagonal, if it is a necessary diagonal and $\tau(i,j) \leq 2\pi/3$.

\begin{lemma}
	\label{lem:BottleneckWithCandidate}
	If there is no bottleneck matching with at most one cascade, then there is a $3$-cascade bottleneck matching $M$, such that at least one inner diagonal of $M$ is a candidate diagonal.
\end{lemma}
\begin{proof}
	Lemma~\ref{lem:BottleneckWithAllNecessary} provides us with a $3$-cascade matching $M$ whose every inner diagonal is necessary. At least one of the $3$ inner diagonals of $M$ has turning angle at most $2\pi/3$, hence it is a candidate diagonal. Otherwise, the total turning angle would be greater than $2\pi$, which is impossible.
\end{proof}

Let us now look at a candidate diagonal $(i,j)$, and examine the position of points $\p{i+1}{j-1}$ relative to it. We construct the circular arc $h$ on the right side of the directed line $v_iv_j$, from which the line segment $v_iv_j$ subtends an angle of $\pi/3$, see Figure~\ref{fig:Polarity}. We denote the midpoint of $h$ with $A$. Points $v_i$, $A$ and $v_j$ form an equilateral triangle, hence we are able to construct the arc $a^-$ between $A$ and $v_i$ with the center in $v_j$, and the arc $a^+$ between $A$ and $v_j$ with the center in $v_i$. These arcs define three areas: $\Pi^-$, bounded by $h$ and $a^-$, $\Pi^+$, bounded by $h$ and $a^+$, and $\Pi^0$, bounded by $a^-$, $a^+$ and the line segment $v_iv_j$, all depicted in Figure~\ref{fig:Polarity}.

\fig{fig:Polarity}{Points $v_{i+1}, \ldots, v_{j-1}$ all lie inside either $\Pi^-$ or $\Pi^+$.}{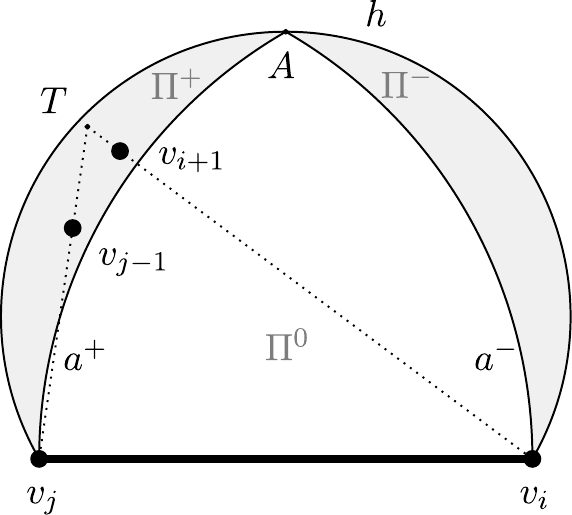}

\begin{lemma}
	\label{lem:Polarity}
	If $(i,j)$ is a candidate diagonal, then points $v_{i+1}, \ldots, v_{j-1}$ either all belong to $\Pi^-$ or all belong to $\Pi^+$.
\end{lemma}
\begin{proof}
	Let $T$ be the point of intersection of lines $v_iv_{i+1}$ and $v_jv_{j-1}$, see Figure~\ref{fig:Polarity}. Since $\tau(i,j) \leq 2\pi/3$, the point $T$ lies in the area bounded by the line segment $v_iv_j$ and the arc $h$. Because of convexity, all points in $\pij$ must lie inside the triangle $\triangle{v_iTv_j}$, so there cannot be two points from $\p{i+1}{j-1}$ such that one is on the right of the directed line $v_iA$ and the other is on the left of the directed line $v_jA$. This as well means that either $\Pi^-$ or $\Pi^+$ is empty. Without loss of generality, let us assume that there are no points from $\p{i+1}{j-1}$ in $\Pi^-$.
	
	It remains to be proved that none of the points in $\p{i+1}{j-1}$ lies in $\Pi^0$. Suppose the opposite, that there is such a point in $\Pi_0$. Let $k$ be the first index in the sequence $\pij$ such that $v_k \in \Pi^+$, see Figure~\ref{fig:PolarityProof}. Since $(i,j)$ is a feasible pair, $\pij$ is of even size, implying that the parity of the number of points in $\Pi^+$ is the same as the parity of the number of points in $\Pi^0$. (If there are points on $a^+$, we assign them to either region.)

	If the number of points in $\Pi^+$, as well as in $\Pi^0$, is odd (not counting points $v_i$ and $v_j$), see Figure~\ref{fig:PolarityProofOdd}, we make a matching using pairs $(i,i+1), (i+2,i+3), \ldots, (j-1,j)$. In the case the number is even, see Figure~\ref{fig:PolarityProofEven}, we make a matching using pairs $(i,i+1), (i+2,i+3), \ldots, (k-3,k-2)$, pair $(k-1,j)$, and pairs $(k,k+1), (k+2,k+3), \ldots, (j-2,j-1)$.
	
	\begin{figure}[ht]
		\centering
		\subfigure[]{
			\label{fig:PolarityProofOdd}
			\includegraphics{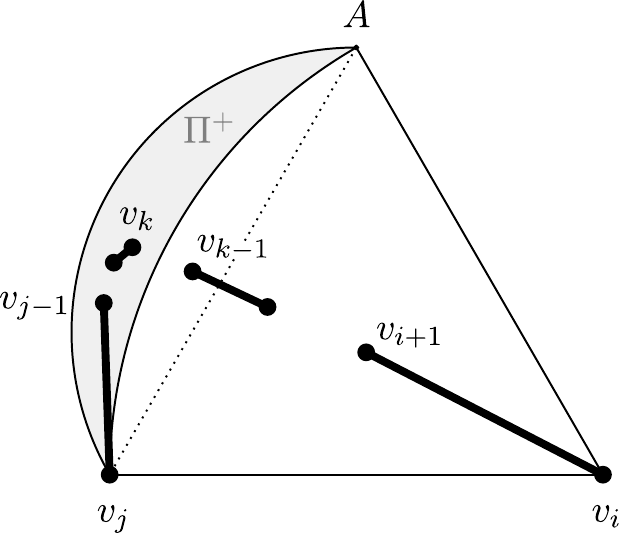}
		}
		\hspace{20pt}
		\subfigure[]{
			\label{fig:PolarityProofEven}
			\includegraphics{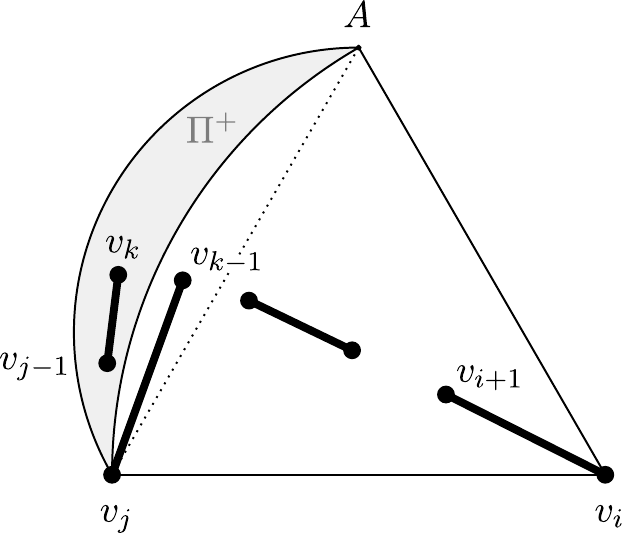}
		}
		\caption{Matching points in $\pij$, depending on the parity of their number.}
		\label{fig:PolarityProof}
	\end{figure}
		
	For each matched pair in any of these two cases, points of the pair either both belong to $\Pi^0$, or they both belong to the area bounded by $a^+$ and the line segment $Av_j$. Both of these areas have diameter $|v_iv_j|$, so all matched pairs have distance not larger than $|v_iv_j|$. So, in each of the cases we constructed a matching $M_{i,j}$ of all points in $\pij$ which does not contain $(i,j)$, with $bn(M_{i,j}) \leq |v_iv_j|$. The matching also satisfies the condition for subproblem $\Matching(i,j)$, i.e.\ it has at most one cascade, and the pair $(i,j)$ belongs to a region bounded by at most one diagonal from $M_{i,j}$ different from $(i,j)$ (which can only be the diagonal $(v_{k-1}v_j)$ in the second case). Consequently, $(i,j)$ cannot be a necessary diagonal, and, thereby, it can not be a candidate diagonal, leading to a contradiction with the assumption that there is a point from $\p{i+1}{j-1}$ in $\Pi^0$.
\end{proof}

With $\Pi^-(i,j)$ and $\Pi^+(i,j)$ we respectively denote areas $\Pi^-$ and $\Pi^+$ corresponding to a candidate diagonal $(i,j)$.

Two possibilities for a candidate diagonal $(i,j)$ provided by Lemma~\ref{lem:Polarity} bring forth a concept of \emph{polarity}. If points $\p{i+1}{j-1}$ lie in $\Pi^-(i,j)$ we say that candidate diagonal $(i,j)$ has \emph{negative polarity} and has $i$ as its \emph{pole}. Otherwise, if these points lie in $\Pi^+(i,j)$, we say that $(i,j)$ has \emph{positive polarity} and pole in $j$.

\begin{lemma}
	\label{lem:CandidatesDontTouch}
	No two candidate diagonals of the same polarity can have the same point as a pole.
\end{lemma}
\begin{proof}
	Let us suppose the contrary, that is, that there are two candidate diagonals of the same polarity with the same point as a pole. Assume, without loss of generality, that $(i,k)$ and $(j,k)$ are two such candidate diagonals, both with positive polarity, each having its pole in $k$. Without loss of generality, we also assume that $j \in \p{i+1}{k-1}$, see Figure~\ref{fig:CandidatesDontTouch}.	
	
	\fig{fig:CandidatesDontTouch}{Two candidate diagonals of equal polarity having the same pole.}{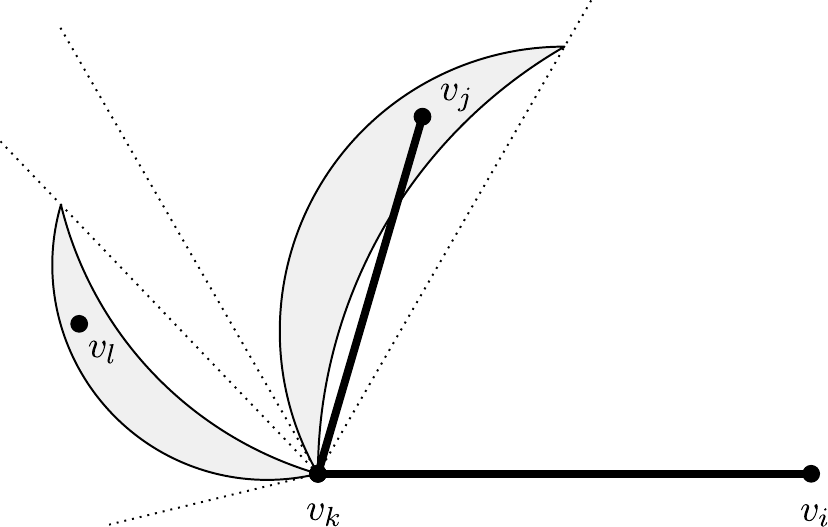}
	
	Area $\Pi^+(i,k)$ lies inside the angle with vertex $v_k$ and sides at angles of $\pi/3$ and $2\pi/3$ with line $v_kv_i$. Similarly, $\Pi^+(j,k)$ lies inside the angle with vertex $v_k$ and sides at angles of $\pi/3$ and $2\pi/3$ with line $v_kv_j$.
		
	Since $(j,k)$ is a diagonal, there is $l \in \p{j+1}{k-1}$.	Points $v_j$ and $v_l$ lie in $\Pi^+(i,k)$ and $\Pi^+(j,k)$, respectively, meaning that $\pi/3 \leq \angle v_iv_kv_j, \angle v_jv_kv_l \leq 2\pi/3$, implying $2\pi/3 \leq \angle v_iv_kv_j + \angle v_jv_kv_l = \angle v_iv_kv_l \leq 4\pi/3$. This means that $v_l$ does not belong to $\Pi^+(i,j)$. However, that is not possible, since $l \in \p{i+1}{j-1}$ as well, so we have a contradiction.
\end{proof}

As a simple corollary of Lemma~\ref{lem:CandidatesDontTouch}, we get that there is at most linear number of candidate diagonals.

\begin{lemma}
	\label{lem:FewCandidates}
	There are $O(n)$ candidate diagonals.
\end{lemma}
\begin{proof}
	Among all candidate diagonals of the same polarity no two can have a pole in the same point of $P$. Therefore, there are at most $n$ candidate diagonals of the same polarity, and, consequently, at most $2n$ candidate diagonals in total.
\end{proof}

Finally, we combine our findings from Lemma~\ref{lem:BottleneckWithCandidate} and Lemma~\ref{lem:FewCandidates}, as described in the beginning of Section~\ref{sec:FindingBottleneckMatching}, to construct Algorithm~\ref{alg:BottleneckMatching}.

\begin{algorithm}[h]
	\caption{Bottleneck Matching}
	\label{alg:BottleneckMatching}
	\begin{algorithmic}
		\State Calculate $S[i,j]$ and $necessary(i,j)$ for all feasible $(i,j)$ pairs, as described in Section~\ref{sec:Subproblems}.
		
		\State $best \leftarrow \min\{S[i+1,i] : i \in \p{0}{n-1}\}$
		
		\For {all feasible $(i,j)$}		
			\If{$necessary(i,j)$ and $\tau(i,j) \leq 2\pi/3$}
				\For {$k \in \p{j+1}{i-1}$ such that $(j+1,k)$ is feasible}
					\State $best \leftarrow \min\{best, \max\{S[i,j], S[j+1,k], S[k+1,i-1]\}\}$
				\EndFor
			\EndIf
		\EndFor
	\end{algorithmic}
\end{algorithm}

\begin{theorem}
	Algorithm~\ref{alg:BottleneckMatching} finds the value of bottleneck matching in $O(n^2)$ time.
\end{theorem}
\begin{proof}
	The first step, calculating $S[i,j]$ and $necessary(i,j)$ for all $(i,j)$ pairs, is done in $O(n^2)$ time, as described in Section~\ref{sec:Subproblems}. The second step finds the minimal value of all matchings with at most one cascade in $O(n)$ time.

	The rest of the algorithm finds the minimal value of all $3$-cascade matchings. Lemma~\ref{lem:BottleneckWithCandidate} tells us that there is a bottleneck matching among $3$-cascade matchings with one inner diagonal being a candidate diagonal, so the algorithm searches through all such matchings. We first fix the candidate diagonal $(i,j)$ and then enter the inner for-loop, where we search for an optimal $3$-cascade matching having $(i,j)$ as an inner diagonal. Although the outer for-loop is executed $O(n^2)$ times, Lemma~\ref{lem:FewCandidates} guarantees that the if-block is entered only $O(n)$ times. The inner for-loop splits $\p{j+1}{i-1}$ in two parts, $\p{j+1}{k}$ and $\p{k+1}{i-1}$, which together with $\pij$ make three parts, each to be matched with at most one cascade. We already know the values of optimal solutions for these three subproblems, so we combine them and check if we get a better overall value. At the end, the minimum value of all examined matchings is contained in $best$, and that has to be the value of a bottleneck matching, since we surely examined at least one bottleneck matching.
\end{proof}

Algorithm~\ref{alg:BottleneckMatching} gives only the value of a bottleneck matching, however, it is easy to reconstruct an actual bottleneck matching by reconstructing matchings for subproblems that led to the minimum value. This reconstruction can be done in linear time.

\bibliographystyle{plain}
\bibliography{bm}

\begin{thebibliography}{10}

\bibitem{abu2015approximating}
A~Karim Abu-Affash, Ahmad Biniaz, Paz Carmi, Anil Maheshwari, and Michiel Smid.
\newblock Approximating the bottleneck plane perfect matching of a point set.
\newblock {\em Computational Geometry}, 48(9):718 -- 731, 2015.

\bibitem{abu2014bottleneck}
A~Karim Abu-Affash, Paz Carmi, Matthew~J Katz, and Yohai Trabelsi.
\newblock Bottleneck non-crossing matching in the plane.
\newblock {\em Computational Geometry}, 47(3):447--457, 2014.

\bibitem{aichholzer2009compatible}
Oswin Aichholzer, Sergey Bereg, Adrian Dumitrescu, Alfredo Garc{\'\i}a, Clemens
  Huemer, Ferran Hurtado, Mikio Kano, Alberto M{\'a}rquez, David Rappaport,
  Shakhar Smorodinsky, Diane Souvaine, Jorge Urrutia, and David~R Wood.
\newblock Compatible geometric matchings.
\newblock {\em Computational Geometry}, 42(6):617--626, 2009.

\bibitem{aichholzer2010edge}
Oswin Aichholzer, Sergio Cabello, Ruy Fabila-Monroy, David Flores-Penaloza,
  Thomas Hackl, Clemens Huemer, Ferran Hurtado, and David~R Wood.
\newblock Edge-removal and non-crossing configurations in geometric graphs.
\newblock {\em Discrete Mathematics and Theoretical Computer Science},
  12(1):75--86, 2010.

\bibitem{alon1993long}
Noga Alon, Sridhar Rajagopalan, and Subhash Suri.
\newblock Long non-crossing configurations in the plane.
\newblock In {\em Proceedings of the ninth annual symposium on Computational
  geometry}, pages 257--263. ACM, 1993.

\bibitem{aloupis2015matching}
Greg Aloupis, Esther~M Arkin, David Bremner, Erik~D Demaine, S{\'a}ndor~P
  Fekete, Bahram Kouhestani, and Joseph~SB Mitchell.
\newblock Matching regions in the plane using non-crossing segments.
\newblock EGC, 2015.

\bibitem{aloupis2013non}
Greg Aloupis, Jean Cardinal, S{\'e}bastien Collette, Erik~D Demaine, Martin~L
  Demaine, Muriel Dulieu, Ruy Fabila-Monroy, Vi~Hart, Ferran Hurtado, Stefan
  Langerman, Maria Saumell, Carlos Seara, and Perouz Taslakian.
\newblock Non-crossing matchings of points with geometric objects.
\newblock {\em Computational geometry}, 46(1):78--92, 2013.

\bibitem{biniaz2014bottleneck}
Ahmad Biniaz, Anil Maheshwari, and Michiel Smid.
\newblock Bottleneck bichromatic plane matching of points.
\newblock Canadian Conference on Computational Geometry, 2014.

\bibitem{carlsson2015bottleneck}
John~Gunnar Carlsson, Benjamin Armbruster, Haritha Bellam, and Rahul Saladi.
\newblock A bottleneck matching problem with edge-crossing constraints.
\newblock {\em International Journal of Computational Geometry and
  Applications}, to appear.

\bibitem{chang1992solving}
Maw-Shang Chang, Chuan~Yi Tang, and Richard C.~T. Lee.
\newblock Solving the euclidean bottleneck matching problem by k-relative
  neighborhood graphs.
\newblock {\em Algorithmica}, 8(1-6):177--194, 1992.

\bibitem{efrat01geometryhelps}
Alon Efrat, Alon Itai, and Matthew~J Katz.
\newblock Geometry helps in bottleneck matching and related problems.
\newblock {\em Algorithmica}, 31(1):1--28, 2001.

\bibitem{efrat2000computing}
Alon Efrat and Matthew~J Katz.
\newblock Computing euclidean bottleneck matchings in higher dimensions.
\newblock {\em Information processing letters}, 75(4):169--174, 2000.

\bibitem{kratochvil2013non}
Jan Kratochv{\'\i}l and Torsten Ueckerdt.
\newblock Non-crossing connectors in the plane.
\newblock In {\em Theory and Applications of Models of Computation}, volume
  7876 of {\em Lecture Notes in Computer Science}, pages 108--120. Springer,
  2013.

\end{thebibliography}

\end{document}